\documentclass[authoryear]{elsarticle}

\usepackage{graphicx}
\usepackage{amsmath,amsthm,amssymb,graphicx,xspace,verbatim}

\usepackage{paralist} 

\theoremstyle{plain}
\newtheorem{thm}{Theorem}
\newtheorem{lem}[thm]{Lemma}
\newtheorem{prop}[thm]{Proposition}
\newtheorem{cor}[thm]{Corollary}

\theoremstyle{definition}
\newtheorem*{dfn}{Definition}

\newcommand \GMk{GM$k$\xspace}
\newcommand \Mk{M$k$\xspace}
\newcommand \Mkv{M$k$v\xspace}
\newcommand \GM{GM\xspace}
\newcommand \GMtwo{GM2\xspace}
\newcommand \GMp{$\text{GM$k$}_{\text{pars-inf}}$\xspace}
\newcommand \GMv{{GM$k$v}\xspace}
\newcommand \GMtwop{$\text{GM2}_{\text{pars-inf}}$\xspace}
\newcommand \GMtwov{$\text{GM2v}$\xspace}
\newcommand \Mtwop{$\text{M2}_{\text{pars-inf}}$\xspace}
\newcommand \CFN{$\text{CFN}$\xspace}
\newcommand \Mkp{$\text{M$k$}_{\text{pars-inf}}$\xspace}
\newcommand \diag{\operatorname{diag}}
\newcommand \proj{\operatorname{proj}}

\newcommand \im{\operatorname{im}}
\newcommand \dist{\operatorname{dist}}

\newcommand{\MkVarNS}{M$k$v}
\newcommand{\MkVar}{{\MkVarNS} }

\usepackage{ifthen}
\usepackage{color}
\newboolean{showeditcomments}
\setboolean{showeditcomments}{true} 
\newcommand{\editcomment}[2]
          {\ifthenelse{\boolean{showeditcomments}}
            {\textcolor{red}{[\textbf{#1:} \emph{#2}]}}
            {}
          }

\usepackage{hyperref}
\hypersetup{backref,  pdfpagemode=FullScreen,  urlcolor=blue, colorlinks=true, citecolor=blue, hyperindex=true}

\begin{document}

\begin{frontmatter}

\title{Estimating Trees from Filtered Data: Identifiability of Models for Morphological Phylogenetics}

\author[UAF]{Elizabeth S. Allman}
\author[MTH]{Mark T. Holder}
\author[UAF]{John A. Rhodes}
\address[UAF]{Department of Mathematics and Statistics, University of Alaska Fairbanks, PO Box 756660, Fairbanks, AK 99775, U.S.A}
\address[MTH]{Department of Ecology and Evolutionary Biology, University of Kansas, 1200 Sunnyside Avenue, Lawrence, Kansas 66045, U.S.A}

\begin{abstract}
As an alternative to parsimony analyses, stochastic models have been proposed \citep{Lewis2001,NylanderRHN2004} 
for morphological characters, so that maximum likelihood or Bayesian analyses may be used for phylogenetic inference. 
A key feature of these models is that they account for ascertainment 
bias, in that only varying, or parsimony-informative characters are observed. 
However, statistical consistency of such model-based inference requires that the model parameters be identifiable from the joint distribution they entail, and this issue has not been addressed.

Here we prove 
that  parameters for several such models, with finite state spaces of arbitrary size, are identifiable, provided
the tree has at least $8$ leaves.  If the tree topology is already known, then $7$ leaves suffice for identifiability of the numerical parameters.  The 
method of proof involves first inferring a full distribution of both parsimony-informative
and non-informative pattern joint probabilities from the parsimony-informative
ones, using phylogenetic invariants.  The failure of identifiability of the tree parameter for $4$-taxon trees
is also investigated.
\end{abstract}


\end{frontmatter}

\section{Introduction}
Currently, the vast majority of phylogenetic inference from morphological data is based on the parsimony criterion -- the preference for the phylogeny that can explain the data with the fewest number of changes in character states.
\citet{Lewis2001} discussed the obstacles to the application of Markov models for phylogenetic inference to discrete morphological data.  
He argued that, despite its limitations, the simplest continuous-time Markov model offers advantages over relying solely on parsimony.
He referred to this model as the \Mk model, for ``Markov'' with $k$-states.
The \Mk model is a generalization of some of the earliest models used in phylogenetics by \citet{JukesC1969}, \citet{Neyman}, \citet{Farris1973}, and \citet{Cavender1978}; it assumes that all states have the same frequency and all transitions between different states occur at the same rate.
Maximum likelihood (ML) inference under the \Mk model should be able to infer trees more accurately than parsimony, because use of the \Mk model allows one to take into account that some branches on the tree may be longer than others.
This branch length heterogeneity could arise from differences in the temporal duration of branches, differences in the rate of character evolution, or both.
Parsimony does not attempt to correct for the fact that convergent changes may not be equally likely to occur on every branch of the tree; as a result, it has been shown to be susceptible to ``long-branch attraction'' when branch-length heterogeneity is present \citep{Felsenstein1978b}.

As \citet{Lewis2001} pointed out, complications arise when applying the \Mk model to  
morphological characters.
The definitions of both characters (also referred to as ``transformation series'') and character states are problematic in morphological systematics.
Systematists disagree about the most appropriate meanings of concepts such as homology \citep[{\em cf.}\;][]{Sereno2007,RieppelK2007,Wiley2008} which are crucial to character coding.
Even if one chooses a particular definition of homology, doing so does not establish a clear-cut set of rules for ``atomizing'' the complete morphology of an organism to a set of characters that can be treated as 
independent instances of a general model of character evolution \citep[but see][for one example of an attempt to create an objective system for character coding]{Ramirez2007}.
While there are rarely clear criteria for delimiting character states or identifying homologous traits between species, systematists 
try to find traits that can be cleanly scored into one of a few discrete bins.
For example, the degree contact of bones in the skull could be scored as a $2$-state character with state 0 representing ``not touching'' and state 1 representing ``in contact.''
If these characters are heritable and do not change too quickly over evolutionary time, then even simple, discrete-state coding scheme can provide information about evolutionary relationships.
Of particular concern here is the fact that variation across the taxa under investigation is vital to the process of recognizing characters and character states.
This means that it is not appropriate to view the coded character matrix as a random sampling of characters
generated by the evolutionary process, since it is biased to contain characters thought to be phylogenetically useful. Such \emph{ascertainment bias}, in which the collected data fails to be representative of the entire population of characters, must be corrected for in a valid statistical analysis.

It is difficult to precisely describe the biases inherent in the process of the coding of morphological traits into columns of a data matrix.
For one thing, it is relatively rare for systematists to even report their methods for excluding potential characters from consideration; \citet{PoeW2000} found that fewer than 20\% of papers in morphological systematics reported such criteria.
The requirement that there be variability among the taxa of interest is clearly one important aspect of character coding.
As noted by \citet{Sereno2007}, several definitions of ``character'' or ``homology'' given by systematists include the idea that a
character differentiates between taxa.  For example he pointed out the following definitions of ``character'' in systematics:
\begin{compactitem}
\item ``Any attribute of an organism or a group of organisms by which it differs from an organism belonging to a different category or resembles an organism of the same category'' \citep[p.~315,][]{MayrLU1953}
\item ``We will call those peculiarities that distinguish a semaphoront (or a group of semaphoronts) from 
other semaphoronts `characters'$\ldots$''\citep[p.~7,][]{Hennig1966} 
\item ``an observation that captures distinguishing peculiarities among organisms $\ldots$'' \citep{RieppelK2002}
\end{compactitem}
The emphasis on variability among taxa means that constant characters generally do not appear in morphological character matrices.
When they do occur, it is often the result of pruning the list of taxa (the characters had been chosen because of variation among members of a larger set of taxa).
As \citet{Lewis2001} noted, this bias cannot be corrected by morphologists changing their systems for encoding characters.
How many constant characters should be encoded to represent a complex feature that is identical across a set of taxa?
The question seems intractable because there are no strict rules about how many aspects of a trait should be coded as independent characters.

The absence of constant characters is the most obvious effect of the ascertainment biases in the coding of morphological data matrices.
\citet{Lewis2001} proposed a corrected model, \MkVarNS, which can be used to calculate a likelihood from a data set conditional on the fact that only variable characters are sampled \citep[see also][]{Felsenstein1992}.

\citet{NylanderRHN2004} further noted that many morphological matrices only contain parsimony-informative characters.
A parsimony-informative character is one which does not have the same parsimony score on every tree.
In order for a character to be parsimony-informative, it must have more than one character state that is shared by multiple taxa.
For instance, if a character `wing shape' for a collection of insects has 3 states, and there is exactly one taxon with shape 2 and one with shape 3, with all others having shape 1, then the character is variable, but not parsimony-informative. If, one the other hand, at least two taxa have shape 1 and at least 2 taxa have shape 2, then the character is parsimony-informative.

If parsimony-noninformative characters are avoided in the process of character coding, then one has data from a smaller set of characters than for the \MkVar model and should 
condition the likelihood calculations based on the parsimony-informative ascertainment bias.
We will refer to this model as \Mkp.  
It was introduced by \citet{NylanderRHN2004}, and implemented 
in the freely available software, MrBayes \citep{RonquistH2003}.

The use of the \MkVar and \Mkp model in phylogenetic inference has grown steadily. 
Dozens of studies using these models have now been published.
Unfortunately, in many cases, authors do not report which form of conditioning is used during analyses, so it is impossible to ascertain the relative frequency of the \MkVar model compared to the \Mkp model.

\medskip

The statistical inconsistency of parsimony as an estimator of phylogenetic trees was one of Lewis's (2001) primary reasons
for proposing the \MkVar model as an improved basis of inference.
However, \citet{Lewis2001} did not prove that ML inference using the \MkVar model is a consistent estimator of the phylogeny.
Nor did \citet{NylanderRHN2004} prove the consistency of ML inference under the \Mkp model.

Recall that statistical consistency of a method of inference under a model means that if data is generated according to the model, then as the amount of data grows, the probability of inferring the correct model parameters (e.g.,  the tree topology and numerical parameters such as edge lengths) approaches 1. There is a standard approach to proving consistency under maximum likelihood \citep{Wald49} that reduces the issue to proving the model has identifiable parameters. This means the crucial step is to show that any two different choices of model parameters lead to a different distribution of data. Identifiability of model parameters is equally essential for their inference in Bayesian analyses, and non-identifiability of parameters that are not the focus of such an analysis may also be problematic \citep{Rannala2002}.

The identifiability of the  \Mk model can be proved through arguments based on an appropriate generalization of the Jukes-Cantor distance. However, the question that we
address is whether we can identify the tree and model parameters when the data is filtered to contain only
variable patterns, or only parsimony-informative patterns. This filtering greatly changes the problem, so that a straightforward modification of the proof for \Mk to \Mkp fails.  One instance of the question of tree identifiability using only parsimony-informative pattern frequencies 
was investigated by \cite{ParsCons}. Although that paper is focused on other issues, in the appendix
it is shown that for the \CFN model on 4-taxon trees, several explicit choices of edge lengths
on different tree topologies can lead to identical distributions of parsimony-informative (and constant)
patterns.

\medskip

Here we demonstrate that the tree topology is identifiable under the \MkVar model under sufficiently broad circumstances to justify its use in data analysis.
While we show that under the \Mkp model, $k\ge 2$, the tree topology is not identifiable for 4-taxon trees, more importantly we establish that the tree topology is identifiable when eight or more taxa are involved.  Moreover, 
if the tree is known, then the branch lengths are identifiable on trees of seven or more taxa.
(The need for seven or eight taxa in these statements may be an artifact of our methods; we do not fully analyze the cases of trees with five, six, or seven taxa.)

Our results are actually valid for more general models, the variable-patterns-only and parsimony-informative-patterns-only versions of the $k$-state general Markov model \GMk, a generalization of the \Mk model in which the transition
probabilities on edges are not constrained to be equal among the different states. 
The identifiability of the tree topology for the unfiltered \GMk  was proven by \citet{Steel1994}, and the identifiability of all numerical parameters for this model by \citet{MR97k:92011}. Identifiability results for a 2-class mixture model of \GMk with invariable sites (GM$k$+I) were investigated by \citet{AllmanR2008}. That model is quite closely related to a variable-patterns-only model, as the invariable class essentially makes direct observation of constant patterns from variable characters impossible. Our results here in fact imply strengthenings of some of the theorems in \citet{AllmanR2008}. Furthermore, since the general Markov model includes as submodels the general time-reversible models with fixed rate-matrices describing the substitution process across the tree, our results apply to the filtered versions of those models as well. 

\smallskip

Interestingly, one implication of this work and that of \citet{ParsCons} which seems not to have been widely noticed, is that the most basic example used to explain phylogenetic inference to students
is actually an example of an intractable problem.
Assuming any model encompassing the \Mk model underlies the data, if we attempt to infer an unrooted four-leaf tree using only those characters that are parsimony-informative, then 
no method of inference can consistently identify the correct tree, even if  given an infinite sample of characters. We establish that each of the three possible binary tree topologies can lead to all possible positive distributions of parsimony-informative patterns, thus strengthening the result of \cite{ParsCons}.

We emphasize that the non-identifiability in this case is not an argument for ignoring the ascertainment bias.
If characters are filtered to contain only parsimony-informative patterns and the ascertainment bias is ignored then
inference can be positively misleading in the sense that \citet{Felsenstein1978b} used the phrase -- the incorrect tree can be preferred with increasing support as the number of characters increases. Indeed, using standard software to perform a maximum likelihood analysis of filtered 4-taxon data under the misspecified \Mk model often results in the erroneous inference of a particular tree topology.
While maximum likelihood inference under the correctly-specified \Mkp model does not prefer any tree topology, it will at least not lead to rejection of the true tree (except when some parsimony-informative patterns do not occur, due to sampling error).

We also note that the \MkVar and \Mkp models may be appropriate in contexts outside of morphological systematics.  
For example, one (admittedly flawed) method for incorporating information from insertion/deletion events (indels) in a molecular sequence analysis is to code the absence or presence of a base as a 0/1 character.  
Because columns without indels are generally not coded, and columns in which all taxa lack a nucleotide are impossible to correctly code, such binary characters should be analyzed under a model that conditions on the variability of the characters. (More appropriate ways of modeling indels  are discussed by \citet{Thorne1991} and \citet{Diallo2007}.) In a similar vein, one of the models included in our analysis, \GMtwov, the variable-patterns-only version of the model GM2, has recently been used for a likelihood analysis of intron loss and gain by \citet{Csuros2007}.

Finally, we emphasize that while establishing identifiability of parameters for a model is essential for its use in statistical inference, there are other important issues that we do not address in this work. In particular, \emph{efficiency} concerns how many characters are needed for inference by a particular method such as maximum likelihood to perform well, and \emph{robustness} concerns how well the method performs on data deviating from the assumed model. Even for unfiltered phylogenetic models these questions have mainly been investigated by simulation, rather than theoretically.

\section{Parsimony-informative models}

The models of sequence evolution we consider are submodels of the general Markov model, with observations restricted to variable or parsimony-informative patterns. In this section, we make this more precise.

By an \emph{$n$-taxon tree $T$}, we mean an unrooted, $n$-leaf, topological phylogenetic tree, with leaves labeled by the taxa $a_i$, $i\in [n]$. We do not assume the tree $T$ is binary; it need not be fully resolved. However, we do assume $T$ has no internal nodes of valence $2$.

The $k$-state \emph{general Markov substitution model}, \GMk, on $T$ is parameterized as follows: First, arbitrarily choose some node of $T$ to be the root. Designating
character states by elements of $[k]=\{1,2,\dots,k\}$ , a row vector $\boldsymbol \pi=(\pi_i)\in[0,1]^k$, with entries summing to 1,  gives  
probabilities of each state $i\in[k]$ occurring at the root. On each edge $e$ of $T$, directed away from the root,
a $k\times k$ Markov matrix $M_e$, with rows summing to 1, gives conditional probabilities of each possible state change occurring on that edge.  We refer to the entries of $\boldsymbol \pi$ and the $M_e$ as the \emph{numerical parameters} of \GMk, in contrast to the \emph{tree parameter} $T$, which is non-numerical.

Throughout, we assume that
\begin{enumerate}
\item[(1)] all entries of $\boldsymbol \pi$ and the $M_e$ are strictly positive, and
\item[(2)]  all $M_e$ are non-singular.
\end{enumerate}
Condition (1) is a biologically natural one, implying that all states and all state transitions can occur. It also ensures that, the probability distribution arising for one choice of the root of $T$ and numerical parameters is identical to one for any other choice of the root, with a corresponding appropriate choice of numerical parameters that are unique up to permutations of character states at internal nodes of $T$ \citep{SSH94,AR03}. This means that identifiability of numerical parameters for \GMk can only be claimed up to the arbitrary choices of the root and orderings of states at internal nodes. Condition (2), which when restricted to continuous-time models is just the requirement that edge lengths be finite, is needed to avoid other sources of non-identifiability (such as a situation in
which all terminal edges have infinite length, so that no information about internal tree structure is
retained in the joint distribution).

\smallskip
Following \citet{Lewis2001}, we use \Mk to denote the submodel of \GMk which assumes a uniform root distribution, $\boldsymbol \pi=(1/k,\dots,1/k)$, and that for each Markov matrix all off-diagonal entries are equal. Thus M4 is also known as the Jukes-Cantor (JC) model, while M2 is the Cavender-Farris-Neyman (\CFN) model. While \cite{Lewis2001}
presents a continuous-time formulation of this model, that is equivalent  to the submodel of the one given here
by making an additional assumption that off-diagonal matrix entries are smaller than diagonal entries.
As our methods are primarily algebraic, we do not focus on the continuous-time formulation.

For the \GMk (or \Mk) model on a fixed $n$-taxon tree $T$, the joint probability distribution of character states at the leaves of
the $T$ can be expressed by polynomial formulas in the entries of $\boldsymbol \pi$ and the $M_e$.
Denote a \emph{pattern} of states at the leaves of a tree by
a vector $\mathbf i =(i_1,i_2,\dots,i_n)=i_1i_2\dots i_n\in[k]^n$,
where the leaf labeled by taxon $a_j$ displays state $i_j$. We use $p_{\mathbf i}$ to denote the probability of observing pattern $\mathbf i$ that arises from a specific model, tree, and numerical parameters.

\medskip

We wish to modify the above models to describe data that is collected only on parsimony-informative patterns. We will not explicitly treat a variable-patterns-only model, as the necessary modifications are straightforward.

Denote the set of  \emph{parsimony-informative patterns} by
\begin{equation*}
\mathcal I=\{ \mathbf i={i_1i_2i_3\cdots i_n} ~|~ i_{j_1}=i_{j_2}\ne i_{j_3}=i_{j_4},\ \text{for some distinct $j_l$}\}. 
\end{equation*}
For fixed $k$, the total number of patterns grows exponentially with $n$, while the number of parsimony-noninformative ones grow only polynomially. Thus the cardinality of $\mathcal I$ grows exponentially with $n$.

Suppose that
from a total
number of $M$ independent, identically distributed characters described by the \GMk model, we may obtain only data counts $n_{\mathbf i}$ for those patterns $\mathbf i\in\mathcal I$.
Since assumption (1) implies $p_{\mathbf i}>0$ for all $\mathbf i$, we have that
$$\mathcal P(n_\mathbf i>0\text{ for some }\mathbf i\in \mathcal I)\to 1\text{ as }M\to\infty.$$
With $N=\sum_{\mathbf i\in \mathcal I} n_{\mathbf i}$, the total count of observed characters, then $N\le M$ and
$\mathcal P(N>0\text)\to 1$ as $M\to \infty$.

If we were able to observe all patterns, including parsimony-noninformative ones, then observed pattern frequencies would be
$\hat p_{\mathbf i}= {n_{\mathbf i}}/M,$ which, by the strong law of large numbers converges to
$ p_{\mathbf i}$
almost surely as $M\to \infty$. However, since $M$ is unknown from data, we cannot compute $\hat p_{\mathbf i}$ directly  for $\mathbf i\in \mathcal I$. We instead introduce the
observed frequencies $$\hat q_{\mathbf i}= \frac {n_{\mathbf i}}N.$$ These are estimators for
conditional probabilities,  $q_{\mathbf i}$, that one observes $\mathbf i$ given that a parsimony-informative pattern is observed. Thus
\begin{equation} q_{\mathbf i}=\mathcal P(\mathbf i ~|~\text{pattern is parsimony-informative})= \frac{ p_{\mathbf i} }p,\label{eq:qi}\end{equation}
where $$p=\mathcal  P(\text{pattern is parsimony-informative})=\sum_{\mathbf i \in \mathcal I } p_{\mathbf i}.$$  Note that
$\hat q_\mathbf i\to q_{\mathbf i}$ almost surely as $M\to \infty$.

We thus define a parameterized model,  \GMp, which gives values of the $q_{\mathbf i}$, $\mathbf i\in \mathcal I$, as a function of the usual \GMk parameters. For any fixed tree, explicit formulas for the $q_{\mathbf i}$ as rational functions of the numerical model parameters are easily obtained.
Restricting to any submodel of \GMk, we similarly obtain a parsimony-informative version of the submodel. For instance, \Mkp denotes the model describing the restriction of observations of the \Mk
model to parsimony-informative patterns.

Similarly, one can define parameterized models \GMv and  \Mkv in which the non-constant patterns can be observed, by conditioning on the variableness of patterns rather than their parsimony-informativeness.
\section{Results}\label{sec:results}

As mentioned,  \cite{ParsCons} showed that from parsimony-informative patterns alone the tree topology is not identifiable for the \CFN (\emph{i.e.}, M2) model on a 4-taxon tree, at least for certain parameter choices. We begin by extending this negative result to models with more character states, and to the full parameter space.

Consider the model \Mkp on a 4-leaf tree $a_1a_2|a_3a_4$. Since there are $3k(k-1)$ parsimony-informative patterns for the $k$-state model, a probability distribution arising from this model is represented by a vector of $3k(k-1)$ probabilities. However,  these vector entries are all the same for patterns of the same form (i.e,  $q_{xxyy}$ is the same for all choices of distinct states $x$ and $y$, etc.). Thus the distribution can be represented by a vector
$$\vec Q=(Q_{xxyy}, Q_{xyyx}, Q_{xyxy}),$$
where $Q_{xxyy}=k(k-1)q_{xxyy}$, etc., so that 
$$Q_{xxyy}+ Q_{xyxy}+ Q_{xyyx}=1.$$
In 3-space,  $\vec Q$ lies on the part of the plane $x+y+z=1$ in the non-negative octant. This set, the \emph{probability simplex} $\Delta$, is an equilateral triangular patch, with corners $(1,0,0)$, $(0,1,0)$, and $(0,0,1)$.

\begin{thm}\label{thm:int}
The set of all probability distributions arising from the model \Mkp  with positive probabilities of a substitution on each edge of the binary tree $a_1a_2|a_3a_4$ is precisely the interior of $\Delta$.
\end{thm}

As the set of probability distributions described in the theorem is independent of the tree topology, we immediately obtain the following.
\begin{cor}\label{thm:4taxa}
Suppose $T$ is a 4-taxon tree. Then the topology of $T$ is not identifiable for the model \Mkp or \GMp, for any $k\ge 2$. 
\end{cor}

\begin{cor} Suppose  data is generated by the \Mkp or \GMp model, on a 4-taxon tree with parameters resulting in a positive probability of observing every parsimony-informative pattern. Then any method of inference
of the tree topology either (a) always returns all three trees, or (b) can be positively misleading. \end{cor}

The proof of Theorem \ref{thm:int} given in Appendix \ref{sec:4taxa} uses explicit calculations and topological arguments.

Note that standard numerical maximum likelihood software generally infers a particular tree topology when 4-taxon data produced by the \Mkp model is analyzed under the misspecified \Mk model. Thus this model misspecification can lead to positively misleading inference.

\smallskip

For larger trees, one might expect that omitting parsimony-noninformative data would result in little loss of information. 
To establish positive results on the identifiability of parameters for the models \GMp and \Mkp, we focus on \GMp, since results about it apply to its submodels.
We separately address the identifiability of the tree topology and identifiability of the numerical model parameters, since the tree topology must be fixed before the numerical parameters are even meaningful.

\begin{thm}\label{thm:treeid} Suppose $n\ge 8$.
Then any $n$-taxon tree topology is identifiable for the \GMp model, and its submodels, such as \Mkp.\end{thm}

Note that we do not claim $n=8$ is the minimal number of taxa ensuring identifiability for either \GMp or \Mkp, either for all $k$ or for any fixed choice. Our method of proof simply does not apply when $n\le 7$.

Since from a distribution for the \GMv model one may compute that of the \GMp model with the same parameters, we immediately obtain the following.

\begin{cor}\label{cor:treeid} Suppose $n\ge 8$.
Then any $n$-taxon tree topology is identifiable for the \GMv model, and its submodels, such as \Mkv.
\end{cor}

The proof of Theorem \ref{thm:treeid} is given in Appendix \ref{sec:tree}, and depends on the construction of phylogenetic invariants for \GMp (\emph{i.e.}, polynomials that vanish
on any joint distribution of patterns for \GMp arising from a fixed tree topology). These invariants are close in spirit to an encoding of the well-known 4-point condition of \cite{Bun}, using the log-det distance \citep{CF87,Steel1994}, but the restriction to parsimony-informative patterns introduces complications.

\medskip

Assuming the tree topology is already known, we next consider the identifiability of numerical parameters. Although our result on identifiability of the tree topology required at least 8 taxa, fewer taxa suffice for our remaining arguments. 

For small trees, though, there are certainly instances of non-identifiability. For instance, in the 4-taxon case, for either of the models \GMtwop or \Mtwop, we cannot have identifiability of numerical parameters. The easiest way to see this is a dimension count: There are only 6 parsimony-informative patterns for \GMtwop on a 4-taxon tree, yet the model has 11 free numerical parameters. However, the continuous parameterization of the model cannot injectively map any full-dimensional subset of $\mathbb R^{11}$  into a 5-dimensional subspace of $\mathbb R^6$. In fact, any distribution arising from the model must arise from infinitely many choices of parameters. Similarly, the model \Mtwop on a 4-taxon tree has 5 free numerical parameters, but up to symmetry there are only 3 parsimony-informative patterns.

In Appendix \ref{sec:nums} we give the outline of the proof of the following, though the work of Appendix \ref{sec:large} is needed to complete the argument.

\begin{thm} \label{thm:7tax} Suppose $n\ge 7$ and $T$ is a known $n$-taxon tree. Then numerical parameters of the model \GMp, and its submodels, such as \Mkp, on $T$ are identifiable, up to 
choice of a root for $T$ and permutation
of the states at the internal nodes of $T$.
\end{thm}

The issue of identifiability of \GM parameters only up to a permutation of states at internal nodes of $T$ is a well-known one \citep{AR03}, arising because the joint distribution gives no information on which hidden state is which.  \citet{MR97k:92011} removed this ambiguity through a biologically-motivated assumption that all Markov matrices have their largest entries in each row appearing on the diagonal.
As permuting the states at internal nodes has the effect of reordering the rows and columns of the Markov matrices,
the highly-structured pattern of entries in the Markov matrices for \Mk enables one to remove the ambiguity even without Chang's assumption. Thus the identifiability of numerical parameters 
for the \Mkp model is, in fact, complete. 

\medskip

For trees with fewer than $7$ taxa, we obtain a slightly weaker result on identifiability of numerical parameters, as stated and proved in Appendix \ref{sec:5taxa}. Although that result is perhaps of less interest for biological application,  we include it as it provides a good introduction to the method of proof of
Theorem \ref{thm:7tax}.
These proofs again depend on phylogenetic invariants, but invariants not for the model \GMp, but rather for  \GMk \citep{ARgm,ARnme}. These invariants lead to algebraic formulas for determining the values of $p_\mathbf i$ for all $\mathbf i\in[k]^n$ from the values of $q_\mathbf i$ for those $\mathbf i\in \mathcal I$. Then the  identifiability of parameters for the \GMk model established by \cite{MR97k:92011}  completes the proof.

An interesting aspect of the work in Appendix \ref{sec:5taxa} is that our arguments for the \Mtwop model on 5-taxon trees 
establish parameter identifiability for generic parameter choices, but
fail under a molecular clock assumption.
Thus what one might consider the simplest assumption actually leads to a more difficult mathematical analysis, due to the symmetries
inherent in it.

\section{Acknowledgements}

The authors thank the Isaac Newton Institute, and the organizers of its Fall 2007 programme in Phylogenetics, for funding enabling the visits where this work was begun. 
MTH thanks the University of Kansas for travel funds to attend the INI. Funds from the National Science Foundation,  grant DMS 0714830,  made these visits possible for ESA and JAR, as well as supported the conclusion of the research.

All authors contributed equally to this work.

\appendix

\section{Non-identifiability of 4-taxon trees}\label{sec:4taxa}

Our proof of Theorem \ref{thm:int} will require the notion of the
the fundamental group of a space, from algebraic topology. As this does not commonly appear in the phylogenetics literature, \citet{Massey} provides a good development for those unfamiliar with it. The arguments in this appendix have little in common with those of the rest of the paper, so readers interested primarily in other results may elect to move on to Appendix \ref{sec:tree}.

\smallskip

Recall $\Delta$ denotes the 2-dimensional probability simplex, as defined in Section \ref{sec:results}.
To simplify some formulas, it will be convenient to represent a vector $\vec Q=(Q_1,Q_2,Q_3)\in \Delta$ by homogeneous coordinates $[Q_1,Q_2,Q_3]$ which are not all zero and are determined only up to rescaling by a non-zero constant. That is, $[Q_1,Q_2,Q_3]=[\lambda Q_1,\lambda Q_2,\lambda Q_3]$ for any $\lambda \ne 0$. Thus $[Q_1',Q_2',Q_3']$ represents $\vec Q$ with $Q_i=Q_i'/(Q_1'+Q_2'+Q_3')$.

Associate to each of the five edges of the tree $a_1a_2|a_3a_4$ a parameter giving the probability of a substitution occurring on that edge,
with  $s_1,s_2,s_3,s_4$ denoting the parameters on pendant edges leading to taxa $a_1,a_2,a_3,a_4$, respectively, and $s_5$ the parameter on the central edge.
Thus the Markov matrix $M_i$ has diagonal entries $1-s_i$ and off-diagonal entries $s_i/(k-1)$.
We focus on the subset of the parameter space
defined by 
$$S=\left \{(s_1,s_2,s_3,s_4,s_5) ~|~s_i\in\left (0,1-1/k \right ) \right \},$$
which corresponds to finite, positive edge lengths.
However, for technical reasons we will also need to consider the extension of the parameterization to the larger set
$$S'=\left \{(s_1,s_2,s_3,s_4,s_5) ~|~s_i\in\left [0,1-1/k\right ),\text{ and either } s_5>0 \text{ or two }s_i>0 \right \}.$$ 
We let $$\phi:S'\to \Delta$$ denote the (extended) parameterization map giving $\vec Q$ as a function of the 5 edge probabilities. 

For any $\epsilon>0$, let $D_\epsilon=\{\vec Q \in \Delta ~|~ \min Q_i < \epsilon\}$  denote an open neighborhood in $\Delta$ of  $\partial \Delta$, the boundary of $\Delta$. We also use $\partial \Delta$ to denote a loop, starting and ending at $[1,0,0]$, parameterizing $\partial \Delta$ in the counterclockwise direction in Figure \ref{fig:simplex}.

\begin{figure}[h]
\begin{center}
\includegraphics[width=2.5in]{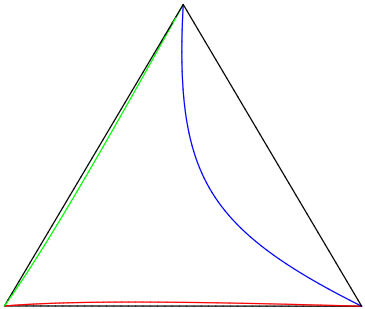}
\end{center}
\vskip -.8in
\caption{The two-dimensional probability simplex $\Delta$, with vertices $[1,0,0]$ (bottom left), $[0,1,0]$ (bottom right), and $[0,0,1]$ (top).
The three curves form the loop $\phi\circ\gamma$ constructed in Lemma \ref{lem:Spgam}, for $k=4$ and $\delta=0.3$, with the image of $\phi\circ\gamma_1$ in red, $\phi\circ\gamma_2$ in blue, and $\phi\circ\gamma_3$ in green. For smaller $\delta$, the loop would be closer to the boundary of $\Delta$.} \label{fig:simplex}
\end{figure}

\begin{lem} \label{lem:Spgam}For any $\epsilon>0$, there exists a loop $\gamma$ in $S'$ such that $\phi\circ\gamma$ is a loop in $D_\epsilon$, starting and ending at $[1,0,0]$, that is homotopic in $D_\epsilon$ to $\partial \Delta$.
\end{lem}

\begin{proof}
We construct the loop $\gamma$ (see Figure 1), in three parts, with $\gamma_1$ chosen so $\phi\circ\gamma_1$ is a path from $[1,0,0]$
to $[0,1,0]$ which is near the edge of $\Delta$ joining those points, $\gamma_2$ chosen so $\phi\circ\gamma_2$ is a path from $[0,1,0]$ to $[0,0,1]$ which is near the edge of $\Delta$ joining those points, and $\gamma_3$ chosen so $\phi\circ\gamma_3$ is a path from $[0,0,1]$ to $[1,0,0]$ which is near the edge of $\Delta$ joining those points. As the construction requires some explicit elementary, but quite long, calculations, we provide these in a worksheet file for the computer algebra software Maple, available as supplementary material on our website \citep{AHRwebsite}.

One can give an explicit formula for $\phi$ (using, for instance, equations (1),(2),(3) of \cite{Schulmeister2004}), and check that 
\begin{align*} 
\phi(0,0,0,s_4,s_5)&=[1,0,0], \text{ for } s_4,s_5\in(0,1-1/k),\\
\phi(s_1,0,0,s_4,0)&=[0,1,0], \text{ for } s_1,s_4\in(0,1-1/k),\\
\phi(0,s_2,0,s_4,0)&=[0,0,1], \text{ for } s_2,s_4\in(0,1-1/k).
\end{align*}

For small $\delta>0$, and for $t\in[0,1]$, let $$\gamma_1(t)=\left ( \frac{2\delta t}{1+2\delta t},0,0,  \frac 13 ,    \frac {\delta(1- t)}{1+\delta (1-t)}\right ),$$
so $\gamma_1(0)=\left ( 0,0,0,  \frac 13 ,    \frac \delta{1+\delta} \right )$ and $\gamma_1(1)=\left ( \frac{2\delta}{1+2\delta},0,0,  \frac 13 ,   0 \right )$,
and one computes
$$\phi\circ \gamma_1(t)=\left [1-t, \frac {t}{k-1}, \frac{(1-t)t\delta}{(k-1)^2}\right ].$$
Note $\phi\circ\gamma_1(0)=[1,0,0]$,  $\phi\circ\gamma_1(1)=[0,1,0]$, and there exists a $\delta_1>0$ such that for all $0<\delta\le\delta_1$  the image of
$\phi\circ \gamma_1$ lies in $D_\epsilon$.

Next, let $$\gamma_2(t)=\left (\frac{ 2\delta(1-t)}{1+2\delta(1-t)},   \frac{ 2\delta t}{1+2\delta t},0,  \frac 13 ,    0\right ),$$
so $\gamma_2(0)=\left ( \frac{2\delta}{1+2\delta},0,0,  \frac 13 ,   0 \right )$ and $\gamma_2(1)=\left ( 0,\frac{2\delta}{1+2\delta},0, \frac 13 ,   0 \right )$.
Then it can be shown that
$$\phi\circ \gamma_2(t)=\left [4(1-t)t\delta, 1-t, t\right ],$$
so $\phi\circ\gamma_2(0)=[0,1,0]$ and  $\phi\circ\gamma_2Ä(1)=[0,0,1]$. Furthermore, there exists a $\delta_2>0$ so that for all $0<\delta\le\delta_2$,  the image of
$\phi\circ \gamma_2$ lies in $D_\epsilon$. 

The third segment of the path is defined similarly to the first,  with
$$\gamma_3(t)=\left ( 0,\frac{2\delta (1-t)}{1+2\delta (1-t)},0,  \frac 13 ,    \frac {\delta t}{1+\delta t}\right ),$$
so $\gamma_3(0)=\left ( 0,\frac{2\delta}{1+2\delta},0,  \frac 13 ,   0 \right )$ and $\gamma_3(1)=\left ( 0,0,0,  \frac 13 ,    \frac \delta{1+\delta} \right )$.  One checks that 
$$\phi\circ \gamma_3(t)=\left [t, \frac{(1-t)\delta t}{(k-1)^2}, \frac{1-t}{k-1}\right ],$$

Then for some $\delta_3>0$, if $\delta\le\delta_3$ then  $\phi\circ\gamma_3$ is a path in $D_\epsilon$ from $[0,0,1]$ to $[1,0,0]$.

Finally, for any $\delta\le\min(\delta_1,\delta_2,\delta_3)$ a loop with the desired properties is given by traversing these paths consecutively, by $\gamma=\gamma_1*\gamma_2*\gamma_3$.\end{proof}

We next obtain a similar result for the parameter space of interest, $S$.
\begin{lem}\label{lem:Sgam}
For any $\epsilon>0$,  there exists a loop $\gamma$ in $S$ such that the loop $\phi\circ\gamma$ is  in $D_\epsilon$ and homotopic  in $D_\epsilon$ to  $\partial \Delta$.
\end{lem}
\begin{proof}
By Lemma \ref{lem:Spgam}, there is a loop $\gamma'$ in $S'$ such that $\phi\circ\gamma'$ is a loop in $D_{\epsilon}$ that is homotopic to  $\partial \Delta$ in $D_\epsilon$. Since $\phi^{-1}(D_\epsilon)$ is open in $S'$ and contains the compact set $\im(\gamma')$, there exists some $\delta'>0$ such that if $\vec s\in S$ and $\dist(\vec s, \im(\gamma'))<\delta'$, then $\phi(\vec s)\in D_\epsilon$. Thus for sufficiently small $\delta>0$, the loop defined by $\gamma(t)=\gamma'(t)+\delta(1,1,1,1,1)$ is in $S$ and $\phi\circ\gamma$ has image in $D_\epsilon$. Since $\gamma$ is homotopic to $\gamma'$ in $\phi^{-1}(D_\epsilon)$, then $\phi\circ\gamma$ is homotopic to $\partial \Delta$ in $D_\epsilon$.
\end{proof}

\begin{proof}[Proof of Theorem \ref{thm:int}]
It is clear that parameters in $S$ lead to positive probabilities of each parsimony-informative pattern, so $\phi(S)\subseteq \Delta\smallsetminus \partial \Delta$.

Let $P\in \Delta\smallsetminus \partial \Delta$, and suppose $P\notin \phi(S)$. Choose $\epsilon>0$ so $P\notin D_\epsilon$, and let $\gamma:[0,1]\to S$ be a loop whose existence is asserted by Lemma \ref{lem:Sgam}. Since a parameterization of $\partial \Delta$ is non-trivial in the fundamental group $\pi_1(\Delta\smallsetminus\{P\})$, $\phi\circ\gamma$ is non-trivial in that fundamental group as well.

However, since $S$ is contractible, there is a homotopy $g$ deforming $\gamma$ to a constant map. Then $h=\phi\circ g$ is a homotopy in $\Delta$ deforming $\phi\circ \gamma$ to a constant map. Since $P\notin \phi(S)$, this is actually a homotopy in $\Delta\smallsetminus \{P\}$. Thus $\phi\circ\gamma$ is trivial in the fundamental group.
This contradiction shows $P\in \phi(S)$. 

Thus $\phi(S)=\Delta\smallsetminus \partial \Delta$. 
\end{proof}

\section{Identifiability of larger trees}\label{sec:tree}

Our argument establishing Theorem \ref{thm:treeid} is at some level similar to ones establishing tree identifiability for more standard models using the existence of a phylogenetic distance. However, because of the filtered nature of the model, we cannot easily define a distance directly.
Instead, we construct certain phylogenetic invariants that can distinguish tree topologies. While these invariants are motivated by a statement of the 4-point condition for the log-det distance, the details of the construction are much more involved.

\smallskip

For proving both Theorem \ref{thm:treeid} and subsequent results, it will be convenient to use the following notation.
Suppose for some choice of parameters for the model \GMk on an $n$-taxon tree the resulting distribution of patterns is given by $\{p_{\mathbf i}\}_{\mathbf i\in [k]^n}$.
Then let $P$ denote the $k\times\dots\times k$ $n$-dimensional array whose entries are $P(i_1,\dots i_n)=p_{i_1\cdots i_n}$. Similarly, for the same parameters
for the model \GMp, suppose the resulting distribution of parsimony-informative patterns is given by
$\{q_{\mathbf i}\}_{\mathbf i\in \mathcal I}$. Then let $Q$ denote a $k\times\dots\times k$ $n$-dimensional array whose entries are $Q(i_1,\dots i_n)=q_{i_1\cdots i_n}$ for $i_1\cdots i_n\in \mathcal I$, and are undefined for $i_1\cdots i_n\notin \mathcal I$. (In this section, we will avoid reference to any undefined entries of $Q$, but in subsequent sections we will give meaning to them.)

\begin{dfn}
Suppose $S$ is some subset of the taxa $\{a_1,\dots, a_n\}$. Then for any pattern $\mathbf i\in [k]^n$, let $\proj_S(\mathbf i)$ denote the vector in $[k]^{|S|}$ of only those components
$i_j$ of $\mathbf i$ with $a_j\in S$. Thus $\proj_S(\mathbf i)$ is the subpattern of $\mathbf i$ of states at the taxa in $S$.
\end{dfn}

\begin{proof}[Proof of Theorem \ref{thm:treeid}] By Theorem 6.3.5 of \citet{MR2060009}, it is enough to show we can identify the topology of the induced subtree for every quartet of taxa. Without loss of generality, we may focus on identifying the topological tree relating  $a_1$, $a_2$, $a_3$, $a_4$.
We may also assume the tree is rooted at the node of our choice: the node where paths leading from $a_1$, $a_2$, and $a_3$ join.

Let $S=\{a_5,\dots,a_n\}$. Choose and fix any pattern $\mathbf i_0=(i_5,i_6,\dots,  i_n)\in[k]^{n-4}$ of states for taxa in $S$ that is
parsimony-informative for $S$. (This requires that $n\ge 8$.) 
Consider the 4-dimensional array $Q_0$ whose entries are all $q_\mathbf i$ such that $\proj_S(\mathbf i)=\mathbf i_0$. This is a 4-dimensional `slice' of the array $Q$ in which only the states at taxa $a_1,a_2,a_3,a_4$ vary. However, $Q_0$ has no undefined entries, as all its entries arise from patterns in $\mathcal I$.

Next we apply the essential idea behind the log-det distance on 4-taxon trees, but modify it to deal with the array $Q_0$. Our argument is similar to that of \citet{Steel1994}, but new details require a full presentation. 

\medskip

Suppose the true quartet tree relating $a_1,a_2,a_3,a_4$ displays the split $a_1a_2|a_3a_4$.
Then to each of the 4 (in the unresolved case) or 5 edges $\tilde e$ of the  quartet tree, we associate a matrix $N_{\tilde e}$ in the following way:

Any edge $\tilde e$ in the quartet tree corresponds to a path $e_1,e_2,\dots, e_r$ in the full tree $T$, possibly with branches leading off toward some of the $a_i$ with $i\ge5$, as illustrated by the representative cartoons
of Figure \ref{fg:Quartets}. 

\begin{figure}[h]
\begin{center}
\includegraphics[height=1in]{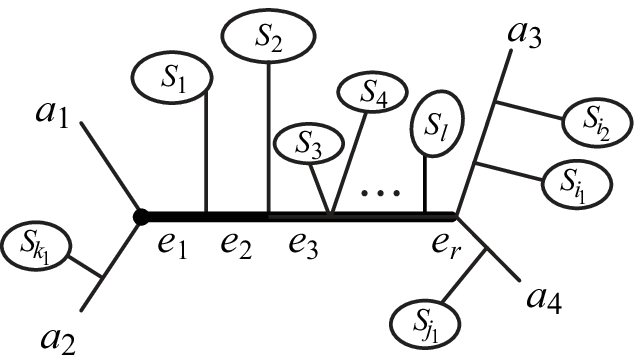} \hskip 1.5cm
\includegraphics[height=1in]{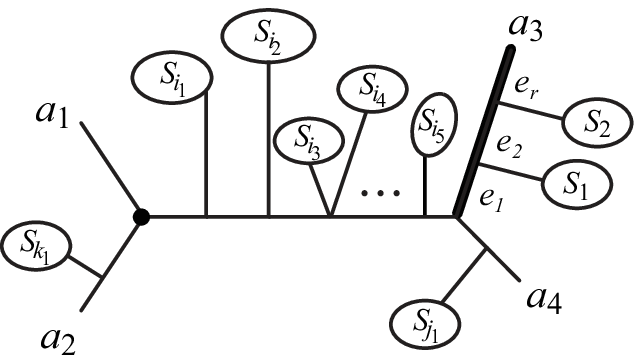}
\end{center}
\caption{Edges $\tilde e$ (shown in bold) in the induced quartet tree correspond to paths $e_1,e_2,\dots, e_r$ in the full tree $T$.
The lines leading off of $\tilde e$ represent the subtrees relating some subcollection of taxa $a_i$, with $i\ge 5$, which are attached in $T$ to nodes along the path.
The common root of  $T$ and the quartet tree is marked with a large dot.}\label{fg:Quartets}
\end{figure}
Consider first a binary tree $T$.
Each subtree of $T$ coming off the path at the node at the end of an $e_i$ contains leaves labeled by taxa in a set $S_i\subseteq S$. To this subtree, associate a vector $\mathbf v_i\in (0,1)^k$ giving the conditional probabilities that each of the states at this node produces the pattern  $\proj_{S_i}(\mathbf i_0)$. While polynomial formulas could be given for these vectors in terms of entries of the Markov matrix parameters, we do not need explicit expressions, so we omit them.   
Now to an edge $\tilde e$ in the quartet tree associate the matrix
\begin{equation}N_{\tilde e}= M_{e_1}\diag(\mathbf v_1)M_{e_2}\diag(\mathbf v_2)M_{e_3}\cdots\diag(\mathbf v_{r-1})M_{e_r},\label{eq:Nprod}\end{equation}
where the $M_{e_i}$ are the Markov matrix parameters on $T$. Thus $N_{\tilde e}$ gives probabilities of changes to all states at the end of $\tilde e$ and to $\proj_{S_i}(\mathbf i_0)$ at the taxa in $S_i$ conditioned on the state at the start of $\tilde e$.

If $T$ is not a binary tree, this expression for $N_{\tilde e}$ is not yet well defined.  By specifying that
subtrees attached to internal nodes of the quartet tree are considered to be attached to  specific pendant quartet tree edges, we remove some ambiguity,
though the expression for $N_{\tilde e}$ for pendant edges may now begin with one or more diagonal matrices, rather than an $M_e$. We also must
allow more than one adjacent diagonal matrix factor in the expression for $N_{\tilde e}$ given in equation \eqref{eq:Nprod} due to multifurcations in $T$ along $\tilde e$.
In case the quartet tree is also not binary, we may for convenience consider a resolved quartet tree and assume the product associated to the internal quartet edge is empty, with $N_{\tilde e}=I$. 
Note that by our assumption that all $M_{e_i}$ have all positive entries, the non-binary quartet tree is the only case in which any $N_{\tilde e}=I$, and otherwise all entries of $N_{\tilde e}$ are positive. 

In all cases,  our hypotheses ensure $N_{\tilde e}$ is non-singular.

\smallskip

Now for the quartet tree associated to the split $a_1a_2|a_3a_4$, let $N_i$, $i=1,2,3,4$ be the four such matrices associated to the edges leading to the leaves, and $N_5$ the matrix associated to the interior edge, as described above.  Redefine the sets $S_i\subseteq S$ to be the set of  taxa $a_i$, $i\ge 5$, which
are in subtrees of $T$ coming off of each of those five quartet edges.  
The entries of the matrices $N_i$ then give conditional probabilities, conditioned on the state at the start of the quartet edge, of observing each state at the end of the
quartet edge and \emph{also} observing $\proj_{S_i}(\mathbf i_0)$. Although their entries are probabilities, the $N_i$ are typically not  Markov matrices, as  entries in each row add to 1 only when $S_i=\emptyset$.

For $i=1,2,3,4$, let $\mathbf w_i=N_i \mathbf 1$ where $\mathbf 1$ is the column vector with all entries 1. The  entries of $\mathbf w_i$, therefore, give the probabilities of observing  $\proj_{S_i}(\mathbf i_0)$, conditioned on the state at the start of the pendant quartet edge, since we are simply marginalizing $N_i$ over the index corresponding to $a_i$.

Let $\mathbf w_{34}$ be the column vector of probabilities of observing $\proj_{S_3\cup S_4\cup S_5}(\mathbf i_0)$ conditioned on the state at the root, so
\begin{align*}
\mathbf w_{34}=N_5 \diag(\mathbf w_3)\diag(\mathbf w_4)\mathbf 1=N_5\diag(\mathbf w_3)\mathbf w_4
=N_5\diag(\mathbf w_4)\mathbf w_3.\end{align*} 
Let $\mathbf w_{12}$ be the vector of probabilities of observing $\proj_{S_1\cup S_2\cup S_5}(\mathbf i_0)$, conditioned on the state at the node
where the quartet edges leading to taxa $a_3,a_4$ join. 
Using Bayes' formula to `reroot' the quartet tree at the second internal node, we similarly find
\begin{align*}\mathbf w_{12}&=\diag(\boldsymbol \pi N_5)^{-1}N_5^T \diag(\boldsymbol \pi)\diag(\mathbf w_1)
\mathbf w_2\\&=\diag(\boldsymbol \pi N_5)^{-1}N_5^T \diag(\boldsymbol \pi)\diag(\mathbf w_2)
\mathbf w_1.\end{align*}
Under our hypotheses, all entries of every $\mathbf w_i$ and $\mathbf w_{ij}$ are positive, as there is a positive conditional probability of every state change occurring on every edge of the full tree.

We now have the following matrix formulas expressing 2-dimensional marginalizations of $Q_0$ in terms of model parameters:
\begin{align*}
Q_0(\cdot,\cdot,+,+)&:=\sum_{i,j\in[k]}Q_0(\cdot,\cdot,i,j)=N_1^T\diag(\boldsymbol \pi)\diag(\mathbf w_{34}) N_2,\\
Q_0(\cdot,+,\cdot,+)&:=\sum_{i,j\in[k]}Q_0(\cdot,i,\cdot,j)=N_1^T\diag(\boldsymbol \pi)\diag(\mathbf w_2) N_5 \diag(\mathbf w_4)N_3,\\
Q_0(\cdot,+,+,\cdot)&:=\sum_{i,j\in[k]}Q_0(\cdot,i,j,\cdot)=N_1^T\diag(\boldsymbol \pi)\diag(\mathbf w_2) N_5\diag(\mathbf w_3)N_4,\\
Q_0(+,\cdot,\cdot,+)&:=\sum_{i,j\in[k]}Q_0(i,\cdot,\cdot,j)=N_2^T\diag(\boldsymbol \pi)\diag(\mathbf w_1) N_5 \diag(\mathbf w_4)N_3,\\
Q_0(+,\cdot,+,\cdot)&:=\sum_{i,j\in[k]}Q_0(i,\cdot,j,\cdot)=N_2^T\diag(\boldsymbol \pi) \diag(\mathbf w_1)N_5\diag(\mathbf w_3)N_4,\\
Q_0(+,+,\cdot,\cdot)&:=\sum_{i,j\in[k]}Q_0(i,j,\cdot,\cdot)=N_3^T\diag(\boldsymbol \pi N_5)\diag(\mathbf w_{12}) N_4.
\end{align*}

These imply \begin{equation}\det(Q_0(\cdot,+,\cdot,+))\det(Q_0(+,\cdot,+,\cdot))-
\det(Q_0(\cdot,+,+,\cdot))\det(Q_0(+,\cdot,\cdot,+))=0.\label{eq:phyinv}\end{equation}
As the left hand side of this equation is a polynomial in the $q_\mathbf i$, $\mathbf i \in \mathcal I$, it is a phylogenetic invariant for the model \GMp. It is analogous the the 4-point distance identity
$d(a_1,a_3)+d(a_2,a_4)=d(a_1,a_4)+d(a_2,a_3)$, and it must vanish on any distribution arising from \GMp in which the induced quartet tree on the first four taxa
displays the split $a_1a_2|a_3a_4$.  
Two invariants similar to that of equation (\ref{eq:phyinv}) can be constructed that will vanish if the quartet tree displays the other possible splits. For the split  $a_1a_3|a_2a_4$ we have
\begin{equation}\det(Q_0(\cdot,\cdot,+,+))\det(Q_0(+,+,\cdot,\cdot))-
\det(Q_0(\cdot,+,+,\cdot))\det(Q_0(+,\cdot,\cdot,+))=0,\label{eq:phyinv2}\end{equation}
and for the split
$a_1a_4|a_2a_3$ 
\begin{equation}\det(Q_0(\cdot,\cdot,+,+))\det(Q_0(+,+,\cdot,\cdot))-
\det(Q_0(\cdot,+,\cdot,+))\det(Q_0(+,\cdot,+,\cdot))=0.\label{eq:phyinv3}\end{equation}

To show that we can use these invariants to identify tree topologies, we need only establish strict inequalities analogous to the distance inequality $d(a_1,a_2)+d(a_3,a_4)<d(a_1,a_3)+d(a_2,a_4)$ which holds provided the central edge of a quartet tree displaying   $a_1a_2|a_3a_4$ has non-zero length. Doing so
would imply that for the fully resolved quartet tree exactly one of the three equations \eqref{eq:phyinv}, \eqref{eq:phyinv2}, and \eqref{eq:phyinv3} can hold.
As the formula for the log-det distance involves a minus sign, we reverse the inequality and, assuming $N_5\ne I$, so all entries of $N_5$ are positive, we seek to show
$$\det(Q_0(\cdot,\cdot,+,+))\det(Q_0(+,+,\cdot,\cdot)>\det(Q_0(\cdot,+,\cdot,+))\det(Q_0(+,\cdot,+,\cdot)).$$
By the expressions for the marginalizations above, this is equivalent to
\begin{multline*}\det(N_1^T\diag(\boldsymbol \pi)\diag(\mathbf w_{34})N_2)\det(N_3^T\diag(\boldsymbol \pi N_5)\diag(\mathbf w_{12})N_4)>\\
\det(N_1^T\diag(\boldsymbol \pi)\diag(\mathbf w_2)N_5\diag(\mathbf w_4)N_3)\times\\
\det(N_2^T\diag(\boldsymbol \pi)\diag(\mathbf w_1)N_5\diag(\mathbf w_3)N_4),\end{multline*}
or, since the $N_i$ and $\diag(\boldsymbol \pi)$ are non-singular,
\begin{multline*}\det(\diag (\boldsymbol \pi N_5))\det(\diag(\mathbf w_{12}))\det(\diag(\mathbf w_{34}))>\\\det(N_5)^2\det(\diag(\boldsymbol \pi)\diag(\mathbf w_{1})\diag(\mathbf w_{2})\diag(\mathbf w_{3})\diag(\mathbf w_{4})),\end{multline*}
or, using the above expressions for the $\mathbf w_{ij}$,
\begin{multline}
\det(\diag (\boldsymbol \pi N_5))\det(\diag(\diag(\boldsymbol \pi N_5)^{-1}N_5^T \diag(\boldsymbol \pi)\diag(\mathbf w_1)
\mathbf w_2))\times\\
\det(\diag(N_5\diag(\mathbf w_3)\mathbf w_{4}))>\\\det(N_5)^2\det(\diag(\boldsymbol \pi)\diag(\mathbf w_{1})\diag(\mathbf w_{2})\diag(\mathbf w_{3})\diag(\mathbf w_{4})).\label{eq:2show}\end{multline}

To establish inequality (\ref{eq:2show}) we will use the following:

\begin{lem}\label{lem:diagin}
Suppose $A$ is a $n\times n$ matrix with positive entries, and the row vector $\mathbf x\in \mathbb R^n$ has positive entries. Then $$\det(\diag (\mathbf x A)) > |\det A|\det(\diag(\mathbf x)),$$ and
$$\det(\diag (A \mathbf x^T) ) > |\det A|\det(\diag(\mathbf x)).$$

\end{lem}
\begin{proof}
We prove the $2\times 2$ case here as an illustration. The general proof can be extracted from \cite{Steel1994}.

With $A=\begin{pmatrix} a&b\\c&d\end{pmatrix},$ $\mathbf x=(x,y)$, since $a,b,c,d,x,y>0$,
the first inequality follows from

$$(ax+cy)(bx+dy)> adxy+bcxy > |ad-bc|xy.$$
The second inequality follows from applying the first to the transpose of $A$.
\end{proof}
Now to establish inequality (\ref{eq:2show}), by applying Lemma \ref{lem:diagin} twice, it is enough to show
\begin{multline*}
\det(\diag (\boldsymbol \pi N_5))\cdot
|\det(\diag(\boldsymbol \pi N_5)^{-1}N_5^T \diag(\boldsymbol \pi)\diag(\mathbf w_1))|
\det(\diag(\mathbf w_2))\cdot\\
|\det(N_5\diag(\mathbf w_3))|
\det(\diag(\mathbf w_{4}))\ge
\\\det(N_5)^2\det(\diag(\boldsymbol \pi)\diag(\mathbf w_{1})\diag(\mathbf w_{2})\diag(\mathbf w_{3})\diag(\mathbf w_{4})).\end{multline*}
After canceling many non-zero determinants appearing on both sides of this inequality, we see it simply states that $1\ge 1$.
\end{proof}

\section{Identifiability of numerical parameters}\label{sec:nums}

The full proof of Theorem \ref{thm:7tax}, on identifiability of numerical model parameters,  depends upon a key technical lemma. This lemma requires extensive arguments that are deferred to Appendix \ref{sec:large}. To motivate the lemma, and make the flow of the larger argument clearer, we first give the proof of the Theorem assuming that lemma is known.

\begin{proof}[Proof of Theorem \ref{thm:7tax}]

For $\mathbf i\in \mathcal I$, $q_{\mathbf i}$ has been defined 
in equation (\ref{eq:qi}), as the conditional probability of observing pattern $\mathbf i$ given that a parsimony-informative pattern is observed.  For mathematical convenience, we extend the definition
of $q_{\mathbf i}$ by the formula in equation (\ref{eq:qi}) to \emph{all} $\mathbf i$, but do not give a probabilistic interpretation to its meaning for  $\mathbf i\notin \mathcal I$. We emphasize that the denominator in this definition remains a sum only over $\mathbf i\in\mathcal I$.

In Appendix \ref{sec:large}, Lemma \ref{lem:large} will show that from the $q_{\mathbf i}$ with $\mathbf i\in \mathcal I$ arising from the \GMp model on a known tree
of at least 7 taxa,  we may determine all
$q_{\mathbf i}$ with $\mathbf i\notin \mathcal I$.  As motivating and proving this lemma requires an extended exposition, we simply assume the result for now.

By equation (\ref{eq:qi}) we know that for $\mathbf i\in [k]^n$ the
$p_{\mathbf i}$ can be obtained from the $q_{\mathbf i}$ by rescaling by the (unknown) factor $p=\sum_{\mathbf i\in \mathcal I} p_\mathbf i$.
Since $\sum_{\mathbf i\in [k]^n} p_{\mathbf i}=1$, however, we may determine $p$ by the formula $p=1/(\sum_{\mathbf i\in[k]^n} q_{\mathbf i})$. Thus
we can determine all $p_{\mathbf i}$ from all $q_{\mathbf i}$.

Finally, with all $p_{\mathbf i}$ known, we can apply the identifiability result of \citet{MR97k:92011} on the \GMk model to complete the argument.
Chang's formulation actually requires additional assumptions on the \GMk model parameters (`diagonal
dominance in rows') which enable one to determine the ordering of the rows and columns of each Markov matrix parameter. As we have not made such an assumption, we note his argument shows the parameters are
only determined up to permutations of states at the internal nodes of the tree.  
\end{proof}

As this proof outline indicates, the major step is in establishing Lemma \ref{lem:large}. Although not logically necessary, to motivate the proof of that lemma, we first investigate the 5-taxon tree case for the model \GMtwop in the next section.
Complications will arise, due to the possibility that certain expressions
may be zero. That will lead us to first establish identifiability for generic parameters in the 5-taxon case,
and then investigate whether exceptional non-identifiable choices of parameters may exist.

\section{Identifiability of numerical parameters: the 5-taxon, \GMtwop case}\label{sec:5taxa}

Following the proof of Theorem \ref{thm:7tax}, to establish identifiability of numerical parameters for the \GMtwop model on a 5-taxon tree, it would be enough to show the $q_{\mathbf i}$ for $\mathbf i\in \mathcal I$ determine
those for $\mathbf i\notin \mathcal I$.
Although we will
see this is not true in complete generality, investigating the conditions under which it is true will
raise some interesting further 
questions, as well as point the way toward Lemma \ref{lem:large}.

\smallskip

We need the following result, a special case of a more general theorem proved by
\citet{ARgm}. (For a more expository presentation, see \citet{ARnme}.)

\begin{thm}\label{thm:2inv} For the \GMtwo model on a 5-taxon  binary tree as
shown in Figure \ref{fig:5taxa}, let $\{0,1\}$ denote the set of character states.
Let $p_{i_1i_2i_3i_4i_5}$ denote the joint
probability of observing state $i_j$ in the sequence at leaf $a_j$,
$j=1,\dots,5$.
\begin{figure}[h]
\begin{center}
\includegraphics[height=1in]{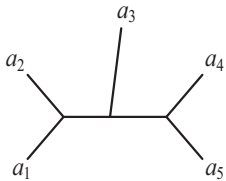}
\end{center}
\caption{A 5-taxon tree}\label{fig:5taxa}
\end{figure}
Then the ideal of phylogenetic invariants for this model
are generated by the $3\times 3$ minors
of the following two matrices:
$$F_1=\begin{pmatrix} p_{00000}& p_{00001}& p_{00010}& p_{00011}& p_{00100}&
p_{00101}& p_{00110}& p_{00111} \\ p_{01000}& p_{01001}& p_{01010}&
p_{01011}& p_{01100}& p_{01101}& p_{01110}& p_{01111}\\ p_{10000}&
p_{10001}& p_{10010}& p_{10011}& p_{10100}& p_{10101}& p_{10110}&
p_{10111}\\ p_{11000}& p_{11001}& p_{11010}& p_{11011}& p_{11100}&
p_{11101}& p_{11110}& p_{11111} \end{pmatrix}$$
and
$$F_2=\begin{pmatrix}
p_{00000}&p_{00001}&p_{00010}&p_{00011}\\
p_{00100}&p_{00101}&p_{00110}&p_{00111}\\
p_{01000}&p_{01001}&p_{01010}&p_{01011}\\
p_{01100}&p_{01101}&p_{01110}&p_{01111}\\
p_{10000}&p_{10001}&p_{10010}&p_{10011}\\
p_{10100}&p_{10101}&p_{10110}&p_{10111}\\
p_{11000}&p_{11001}&p_{11010}&p_{11011}\\
p_{11100}&p_{11101}&p_{11110}&p_{11111} \end{pmatrix}.  $$

\end{thm}

\smallskip

A few comments may make this theorem clearer.
The matrices $F_1$ and $F_2$ are the two natural 2-dimensional `flattenings' of the 5-dimensional joint distribution array according to the
splits corresponding to the two internal edges of the tree. The splits, are
$\{\{a_1,a_2\},\{a_3,a_4,a_5\}\}$, and
$\{\{a_1,a_2,a_3\},\{a_4,a_5\}\}$, and the indices of the matrix entries are such that
states are held constant in one of these sets as one moves across rows or down columns.

Recall that a $3\times 3$ minor of a matrix is defined as the determinant of
a $3\times 3$ submatrix obtained by deleting all but 3 rows and all but 3 columns.
Thus each of these matrices has $4\binom 83=224$ such minors.
Saying these 448 polynomials are phylogenetic invariants means that they evaluate to 0 on any distribution
arising from the model. We
view each of these polynomials 
as specifying an algebraic relationships between the various $p_\mathbf i$.

Of course these relationships imply algebraic relationships between the $q_\mathbf i$ as well.

\begin{cor}
Every $3\times 3$ minor of the two matrices $\widetilde F_1,\widetilde F_2$ obtained from $F_1,F_2$ by replacing all $p_{\mathbf i}$ by $q_{\mathbf i}$ equals zero, if the $q_{\mathbf i}$ arise from the \GMtwo model on the 5-taxon tree.
\end{cor}
\begin{proof}
Since the matrices with entries $q_{\mathbf i}$ are simply rescalings of those
with entries $p_{\mathbf i}$, this follows from the fact that determinants are homogeneous polynomials.
\end{proof}

Thus we know many algebraic relationships between the $q_\mathbf i$. We
now exploit these to determine the $q_\mathbf i$, $\mathbf i\notin \mathcal I$ from the
$q_\mathbf i$, $\mathbf i\in \mathcal I$.

\medskip

Consider first the matrix $\widetilde F_1$, where we use an underscore, as in `$\underline{ q_\mathbf i}$', to highlight those entries where $\mathbf i\notin \mathcal I$ (\emph{i.e.}, the entries we wish to determine).

$$\widetilde F_1=\begin{pmatrix} \underline{ q_{00000}}&\underline{ q_{00001}}&\underline{q_{00010}}& q_{00011}&\underline{ q_{00100}}&
q_{00101}& q_{00110}& q_{00111} \\\underline{q_{01000}}& q_{01001}& q_{01010}&
q_{01011}& q_{01100}& q_{01101}& q_{01110}&\underline{ q_{01111}}\\\underline{ q_{10000}}&
q_{10001}& q_{10010}& q_{10011}& q_{10100}& q_{10101}& q_{10110}&
\underline{q_{10111}}\\ q_{11000}& q_{11001}& q_{11010}&\underline{q_{11011}}& q_{11100}&
\underline{q_{11101}}& \underline{q_{11110}}& \underline{q_{11111}} \end{pmatrix}$$
Focusing on the minor using rows 2,3,4 and columns 2,3,4, we find
$$\begin{vmatrix}  q_{01001}& q_{01010}&
q_{01011}\\
q_{10001}& q_{10010}& q_{10011}\\  q_{11001}& q_{11010}&\underline{ q_{11011}}\end{vmatrix}=0.$$
Expanding the determinant in cofactors by the last column we have
$$q_{01011}\begin{vmatrix}
q_{10001}& q_{10010}\\  q_{11001}& q_{11010}\end{vmatrix} -
q_{10011}
\begin{vmatrix}  q_{01001}& q_{01010}\\
 q_{11001}& q_{11010}\end{vmatrix}+
\underline{ q_{11011}} \begin{vmatrix}  q_{01001}& q_{01010}\\
q_{10001}& q_{10010}\end{vmatrix}=0.$$
Thus,
provided $$\begin{vmatrix}  q_{01001}& q_{01010}\\
q_{10001}& q_{10010}\end{vmatrix}\ne 0,$$ we can express $\underline{q_{11011}}$ in terms of only $q_\mathbf i$ with $\mathbf i\in \mathcal I$.
Assuming the non-vanishing of this $2\times 2$ minor, then, we see
$\underline{ q_{11011}}$ is determined by the $q_\mathbf i$ for $\mathbf i\in \mathcal I$.
More generally, as long as any one of the three $2\times 2$ minors built from rows 2,3 and two of the columns 2,3,5 are non-zero, a similar argument shows $\underline{ q_{11011}},$
$\underline{q_{11101}},$  and $\underline{q_{11110}}$ can all be determined. Note that the non-vanishing
of at least one of these minors is equivalent to the condition that the $\{2,3\}$-$\{2,3,5\}$ submatrix
$$L_1=\begin{pmatrix}
q_{01001}& q_{01010}&q_{01100}\\
q_{10001}& q_{10010}&q_{10100}
\end{pmatrix}$$ has rank 2.

We similarly see that provided the $\{2,3\}$-$\{4,6,7\}$ submatrix $$L_2=\begin{pmatrix}  q_{01011}&q_{01101}& q_{01110}\\
q_{10011}&q_{10101}& q_{10110}\end{pmatrix}$$ has rank 2, then $\underline{ q_{00001}},$ $\underline{ q_{00010}},$ and $\underline{ q_{00100}}$ are also determined.

We now consider the other matrix,
$$\widetilde F_2=\begin{pmatrix}
\underline{ q_{00000}}&\underline{ q_{00001}}&\underline{ q_{00010}}&q_{00011}\\
\underline{q_{00100}}&q_{00101}&q_{00110}&q_{00111}\\
\underline{q_{01000}}&q_{01001}&q_{01010}&q_{01011}\\
q_{01100}&q_{01101}&q_{01110}&\underline{q_{01111}}\\
\underline{q_{10000}}&q_{10001}&q_{10010}&q_{10011}\\
q_{10100}&q_{10101}&q_{10110}&\underline{q_{10111}}\\
q_{11000}&q_{11001}&q_{11010}&\underline{q_{11011}}\\
q_{11100}&\underline{q_{11101}}&\underline{q_{11110}}&\underline{q_{11111}} \end{pmatrix}.  $$

Provided its $\{2,3,5\}$-$\{2,3\}$ and $\{4,6,7\}$-$\{2,3\}$ submatrices $$L_3=\begin{pmatrix} q_{00101}&q_{00110}\\
q_{01001}&q_{01010}\\
q_{10001}&q_{10010}\end{pmatrix} \text{ and }
L_4=\begin{pmatrix} q_{01101}&q_{01110}\\q_{10101}&q_{10110}\\
q_{11001}&q_{11010}\end{pmatrix}
$$ also have rank 2 we similarly can
determine $\underline{q_{00000}},$ $\underline{q_{01000}},$ $\underline{q_{10000}},$
$\underline{q_{10111}},$ $\underline{ q_{01111}},$ and $\underline{q_{11111}}$. Note that for
the determination of $\underline{ q_{00000}}$ and $\underline{ q_{11111}}$ we need values of some of the $\underline{q_\mathbf i}$ that have already been determined.

We've thus established
\begin{lem}\label{lem:qi}
Provided all 4 of the matrices

$$L_1=\begin{pmatrix}
q_{01001}& q_{01010}&q_{01100}\\
q_{10001}& q_{10010}&q_{10100}
\end{pmatrix},\ L_2=\begin{pmatrix}  q_{01011}&q_{01101}& q_{01110}\\
q_{10011}&q_{10101}& q_{10110}\end{pmatrix}$$

$$L_3=\begin{pmatrix} q_{00101}&q_{00110}\\
q_{01001}&q_{01010}\\
q_{10001}&q_{10010}\end{pmatrix} \
L_4=\begin{pmatrix} q_{01101}&q_{01110}\\q_{10101}&q_{10110}\\
q_{11001}&q_{11010}\end{pmatrix}
$$
have rank 2, then the $q_\mathbf i$, $\mathbf i\in \mathcal I $ determine all $q_\mathbf i$, $\mathbf i\in [k]^n$.
\end{lem}

\medskip

Combined with the argument like that for Theorem \ref{thm:7tax}, this lemma may be used to quickly establish that numerical parameters are
generically identifiable for both the \GMtwop and \Mtwop models on 5-taxon trees. \emph{Generic identifiability} means that the subset of parameter space for which identifiability may not hold is of measure zero within the full parameter space.
By Lemma \ref{lem:qi}, numerical parameter identifiability may fail only when at least one of the four matrices has rank $\le 2$, a condition which can be equivalently phrased in terms of the vanishing of a finite set of 
polynomials in the $q_{\mathbf i}$, obtained as certain products of  $2\times 2$ minors of the $L_i$. Composing these polynomials with
the polynomial parameterization map for the model, we find the set of all non-identifiable parameter choices lies within the zero set of a finite set of polynomials, \emph{i.e.}, it lies within an algebraic variety. Exhibiting  a single choice of parameters for which these matrices all have rank 2, then, will establish that this is a
proper subvariety of parameter space, and hence is of lower dimension than the full parameter space,
with Lebesgue measure zero. 
Though we omit presenting such an example  here, it is easy to choose rational parameter values and calculate with exact arithmetic to establish that such examples exist.

\medskip

We next investigate for what parameters any of the matrices $L_i$ of Lemma \ref{lem:qi}
has rank $<2$. This will establish generic identifiability in another way, by giving an
explicit characterization of those parameters for which identifiability might not hold. 
Although our analysis will not give
complete understanding of all cases, we show that
while generic parameters are identifiable, there are indeed cases of \GMtwop
parameters that are not identifiable.

\medskip

Consider first the submatrix
$$L_1=\begin{pmatrix}
q_{01001}& q_{01010}&q_{01100}\\
q_{10001}& q_{10010}&q_{10100}
\end{pmatrix},$$ and root the tree at the internal node closest to $a_1$ and $a_2$ in Figure \ref{fig:5taxa}. We use
$M_i$ for the Markov matrix on the terminal edge to $a_i$, $M_6$ for the Markov matrix on the
internal edge leading from the root, and $M_7$ for the Markov matrix on the other internal edge.
Let $$C_1=\begin{pmatrix} M_1(0,0)M_2(0,1)&M_1(0,1)M_2(0,0)\\M_1(1,0)M_2(1,1)&M_1(1,1)M_2(1,0)\end{pmatrix},$$
and
$$C_2=\begin{pmatrix} M_4(0,0)M_5(0,1)&M_4(0,1)M_5(0,0)\\M_4(1,0)M_5(1,1)&M_4(1,1)M_5(1,0)\end{pmatrix}.$$

Then \begin{equation}
L_1=C_1^T\diag( \boldsymbol \pi)M_6 D_1,\label{eq:N1}\end{equation} where
$$D_1=\begin{pmatrix} \mathbf b_1&\mathbf b_2&\mathbf b_3\end{pmatrix},$$
is a $2\times 3$ matrix with columns $\mathbf b_i$ given by
$$\begin{pmatrix} \mathbf b_1&\mathbf b_2\end{pmatrix}=\diag(M_3(\cdot,0))M_7 C_2$$ and
$$\mathbf b_3=\diag(M_3(\cdot,1))M_7 \begin{pmatrix} M_4(0,0)M_5(0,0)\\M_4(1,0)M_5(1,0)\end{pmatrix}.$$
(Here $M(\cdot, i)$ denotes the $i$th column of $M$.)

Thus the first two columns of $L_1$ are given by $$C_1^T\diag( \boldsymbol \pi)M_6\diag(M_3(\cdot,0))M_7 C_2.$$ Note all matrices in this product have rank 2 except possibly the $C_i$. Thus if both $C_i$ have rank 2, so does $L_1$.

A similar argument applies to the other $L_i$, yielding the following explicit statement of
generic identifiability

\begin{thm} \label{thm:genid+} The model \GMtwop has identifiable numerical parameters for all parameter values such that both $C_1$ and $C_2$ have rank 2. \end{thm}

We now investigate under what circumstances the $C_i$ fail to be of rank 2.
With $$M_1=\begin{pmatrix} 1-a_1&a_1\\a_2&1-a_2\end{pmatrix}, \
M_2=\begin{pmatrix} 1-b_1&b_1\\b_2&1-b_2\end{pmatrix},$$ where $0<a_i,b_i<1$,
$$C_1=\begin{pmatrix} (1-a_1)b_1&a_1(1-b_1)\\a_2(1-b_2)&(1-a_2)b_2\end{pmatrix}.$$
Thus  $\det C_1=0$ means $(1-a_1)(1-a_2)b_1b_2=a_1a_2(1-b_1)(1-b_2)$, so
$$\frac {a_1a_2}{(1-a_1)(1-a_2)}=\frac {b_1b_2}{(1-b_1)(1-b_2)}.$$
Letting $\alpha_i= \frac {a_i}{1-a_i}$ and $\beta_i=\frac {b_i}{1-b_i}$, then $0<\alpha_i,\beta_i<\infty$ and these are in  1-1 correspondence with $a_i,b_i$. We now have

\begin{lem}\label{lem:badp}
The matrix $C_1$ has rank 1 if, and only if, $\alpha_1\alpha_2=\beta_1\beta_2$.
\end{lem}

Thus to find examples where $C_1$ has rank 1 we may pick $M_1$ (equivalently $\alpha_1, \alpha_2$, or $a_1,a_2$) arbitrarily, and then have only one free parameter to pick $M_2$ (equivalently, we may pick $\beta_1$ or $b_1$, and then $\beta_2$ and $b_2$ are determined).

If we avoid such `bad' parameter choices for both the Markov matrices on the cherry of taxa 1 and 2 and
the Markov matrices on the cherry of taxa 4 and 5, then \GMtwop has identifiable parameters.

\begin{cor}\label{cor:badp}
Numerical parameters of the model \GMtwop on the 5-taxon tree are identifiable except possibly on a codimension 1
algebraic subvariety of parameter
space. This subvariety is the union of 2 irreducible varieties, one is explicitly characterized by the condition of Lemma \ref{lem:badp} on the Markov matrices $M_1,M_2$, and the other by a similar condition on $M_4,M_5$.
\end{cor}

\medskip

We next investigate whether identifiability actually fails for the parameter choices
indicated in the corollary, or if it is only our proof that fails.

Consider the extreme case where $M_1,M_2,M_4,M_5$ have been chosen so that both $C_1$ and $C_2$ have rank 1. Then from an expression similar to equation (\ref{eq:N1}), the fact that $C_1$ has rank 1
implies that the middle two
rows of the matrix $F_1$, and hence of $\widetilde F_1$, must be dependent.
Thus if we knew the second row of $\widetilde F_1$, and one of the entries in the third row, we could determine
the rest of the third row. Similar comments apply to the middle two columns of $\widetilde F_2$, using that $C_2$
is of rank 1.

This observation shows that if we project from the 20 coordinates $\{q_{\mathbf i}\}_{\mathbf i\in \mathcal I}$ to the 12 coordinates shown in the array
$$\begin{pmatrix} -&-&-& q_{00011}&-&
q_{00101}& *& q_{00111} \\-& q_{01001}& q_{01010}&
q_{01011}& q_{01100}& q_{01101}& *&-\\-&
q_{10001}& *&*& *& *& *&
-\\
q_{11000}& q_{11001}& *&-& q_{11100}&
-& -& - \end{pmatrix}$$
obtained by deleting entries in $\widetilde F_1$, then this projection will be injective on distributions
arising from \GMtwop parameters for which both $C_1$ and $C_2$ have rank 1. In the above array `$-$' marks parsimony-noninformative entries, and `$*$' parsimony-informative ones that can be inferred from other entries shown under the assumption that $C_1$ and $C_2$ have rank 1. 
To establish that \GMtwop is not identifiable for all parameters,
it is thus enough to argue that if we know $C_1$ and $C_2$ have rank 1, identifiability of parameters is impossible from these 12
coordinates. 

Note that the restricted parameter space for the \GMtwop model where $C_1,C_2$ have rank 1 has
dimension 13: the sum of $2\cdot2-1=3$ parameters for each cherry, 2 parameters for each of the 3 other edges, and 1 parameter for the root distribution.
Thus each 13-dimensional neighborhood of a point in the interior of the restricted parameter space
has an image that is of dimension at most 12. Thus the parameters cannot be identifiable, as the map is
infinite-to-one.

\begin{prop}
There exist distributions arising from the \GMtwop model on a 5-taxon tree with infinite fiber under the parameterization map.
That is, infinitely many choices of parameters can lead to the same distribution.
\end{prop}

We now use our earlier theorems, which have all concerned the model \GMtwo, to deduce results on
the model
\Mtwop. 

To specialize  Corollary \ref{cor:badp} to \Mtwop,
note that the condition of Lemma \ref{lem:badp} simplifies to $M_1=M_2$ for this model. Thus we 
obtain the following.

\begin{cor}
For the \Mtwop model on the 5-taxon tree, suppose $M_1\ne M_2$ and $M_4\ne M_5$. Then the
numerical parameters are identifiable.
\end{cor}

Rather interestingly, in the case of
a molecular clock assumption, with a root located anywhere on the tree, the potential bad cases in the statement above, $M_1=M_2$ or $M_4=M_5$, actually arise. It is an open question whether identifiability actually fails for \Mtwop in such cases. This underscores that what may appear to be the simplest biological assumptions may well lead to undesirable mathematical behavior, due to special symmetries.

\section{Identifiability of numerical parameters: large trees}\label{sec:large}

We turn now to establishing Lemma \ref{lem:large}, the key technical point needed in the proof of Theorem \ref{thm:7tax}. 
While the method of proof of is similar to what appears in Appendix \ref{sec:5taxa}, we generalize to models with an arbitrary number of states,
and deal with larger trees in order to avoid obtaining a theorem that only holds for generic parameters. This complicates the presentation, but introduces few new ideas.

\smallskip

We require some additional terminology.

\begin{dfn}
A binary tree is said to have an \emph{$(m,n)$ split} if deleting one edge partitions the taxa into sets of size $m$ and $n$ according to connected components of the resulting graph. A non-binary tree
is said to have an $(m,n)$ \emph{split} if some binary resolution of it does.
\end{dfn}

\begin{lem} $T$ has at least $7$ taxa if, and only if, $T$ has a $(m,n)$ split with $m\ge 4$ and $n\ge 3$.
\end{lem}

\begin{proof} We may assume $T$ is binary. Suppose first $T$ has at least $7$ taxa. We consider three cases based on the number of cherries in $T$.

If $T$ has exactly two cherries, then $T$ is a caterpillar tree and the forward implication is clear.

If $T$ has exactly three cherries, then $T$ is obtained by grafting one
or more additional edges to interior edges of the tree $((a,b),(c,d),(e,f))$ and the forward implication is again clear.

If $T$ has four or more cherries, then $T$ is obtained by grafting rooted trees to the tree
$(((a,b),(c,d)),((e,f),(g,h))$ and the forward implication is clear.

The converse is clear.
\end{proof}

We use this to prove the lemma which is the key ingredient of Theorem \ref{thm:7tax}.

\begin{lem}\label{lem:large} Suppose $T$ is an $n$-taxon tree with $n\ge 7$. Then the $q_\mathbf i$ for $\mathbf i\in \mathcal I$ arising from some choice of \GMp parameters on $T$ uniquely determine
the $q_\mathbf i$ for $\mathbf i\notin \mathcal I$.
\end{lem}

\begin{proof} We may assume $T$ is binary by passing to a binary resolution of it, noting that
the probability distributions arising from the model on the unresolved tree also arise from
the model on the resolved tree by setting Markov matrices on new edges to the identity matrix.

Let $e$ denote some edge of $T$ corresponding to an $(m,n-m)$ split with $m\ge 4$, $n-m\ge 3$.

Recall, the more general version of Theorem \ref{thm:2inv}  
for \GMk on $n$-taxon binary trees \citep{ARgm}: If $P$ denotes the $n$-dimensional $k\times k\times\cdots \times k$
joint distribution tensor with entries $p_{\mathbf i}$, where $\mathbf i$ denotes a pattern, let
$F_e$ be the matrix obtained by flattening
$P$ along $e$. Then all $(k+1)\times(k+1)$ minors of $F_e$ are zero.

Replacing each $p_\mathbf i$ in $F_e$ by $q_\mathbf i$ to obtain a matrix $\widetilde F_e$ preserves the vanishing of these minors, due to the homogeneity of determinants.

For each parsimony-noninformative pattern $\mathbf i\notin \mathcal I$, we will produce a $(k+1)\times(k+1)$ submatrix of $\widetilde F_e$ that involves $ q_\mathbf i$ but no other unknown $ q_\mathbf j$. We will furthermore ensure that the $k\times k$ minor of this submatrix that uses rows and columns complementary to those of $ q_\mathbf i$ is non-zero. Then the vanishing of the $(k+1)\times(k+1)$ 
determinant
leads to a formula for $ q_\mathbf i$ in terms of known $q_\mathbf j$, as in Section \ref{sec:5taxa}. Thus we
may recover all unknown values of $ q_i$ $\mathbf i\notin S$.

To produce these $(k+1)\times(k+1)$  submatrices, we must fix additional notation. With $e$ the fixed edge described above,
we may assume our taxa are labeled so that the partition of taxa induced by removing $e$ has sets $S_1=\{a_1,\dots, a_m\}$ and  $S_2=\{a_{m+1},\dots,a_{n}\}$, so $\widetilde F_e$ has rows indexed by $[k]^m$ and columns by $[k]^{n-m}$. We may further assume taxa $a_{m-1}$ and $a_m$ form a cherry, as do $a_{n-1}$ and $a_{n}$, and the other taxa in $S_1$ are numbered in a manner consistent with the diagram of the subtree of $T$ shown in Figure \ref{fig:halftree}, and similarly for those taxa in $S_2$. Thus taxa are numbered in order of where the path
from the deleted edge $e$ to the taxa leaves the path from the deleted edge to $a_m$ (respectively $a_n$).

\begin{figure}[h]
\begin{center}
\includegraphics[height=1in]{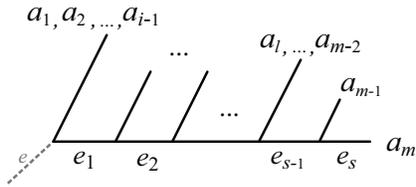}
\caption{Assumed ordering of taxa in the subtree of $T$ to one side of $e$.  }\label{fig:halftree}
\end{center}
\end{figure}

For any pattern $\mathbf i\in[k]^{n}$, let $\mathbf i_1=\proj_{S_1}(\mathbf i)\in[k]^m$ and
$\mathbf i_2=\proj_{S_2}(\mathbf i)\in[k]^{n-m}$. 

The values of $q_\mathbf i$ are known if $\mathbf i$ has among its components at least 2 states that appear at least twice each. 
In cases 1-4 below, we will use these to first determine
those $q_\mathbf i$ for which $\mathbf i$ has exactly one component that appears at least twice, but
$\mathbf i$ is not a constant pattern. 
Without loss of generality, we may assume the component that
appears at least twice in $\mathbf i$ is 1,  yet $\mathbf i\ne (1,1,\dots,1).$

\smallskip
\noindent
\emph{Case 1: No $1$ appears in $\mathbf i_1$, so at least two $1$s appear in $\mathbf i_2$.} All components of $\mathbf i_1$ must be distinct, so let $a\ne b$ be two of these.
Consider the row indices
$$\mathbf i_1,\text{ and for each $i\in[k]$, } \mathbf j_i=(a,a,\dots,a,b,i), $$
and the column indices
$$\mathbf i_2,\text{ and for each $i\in[k]$, }  \mathbf k_i=(a,a,\dots,a,b,i).$$
Then the $(k+1)\times(k+1)$ submatrix of $\widetilde F_e$ formed by these rows and columns has all known entries
except $q_\mathbf i$.  

We further claim the $k\times k$ submatrix $L$ with entries $q_{(\mathbf j_i,\mathbf k_j)}$, $i,j\in[k]$ has non-zero determinant. To see this, note that by viewing the tree $T$ as rooted at the end of $e$ closest to taxon $a_1$, $L$ has a matrix factorization
\begin{equation}
L=C_1^T\diag(\boldsymbol \pi)C_2,\label{eq:factk}
\end{equation}
where the entries of $C_1$ give probabilities
of producing the patterns $\mathbf j_i$ at the taxa in $S_1$ conditioned on the root state, and the entries of $C_2$
similarly give conditional probabilities of producing the patterns $\mathbf k_i$ at the taxa in $S_2$.
Referring to Figure \ref{fig:halftree}, we find
$$C_1=D_1 M_{e_1}D_2\dots D_{s-1}M_{e_{s-1}}D_sM_{e_s},$$
where each $D_i$ is a diagonal matrix whose entries give the probabilities of states
at the $i$th node along the path from the root to $m$ producing the particular pattern $\proj_{B_i}(a,\dots a,a,b)$ on the taxa in the set $B_i$ labeling
the leaves on the subtree branching off from that node. By our assumptions on parameters, all matrices in this product are non-singular, so $C_1$ is as well. A similar product shows $C_2$ is also non-singular, so by equation (\ref{eq:factk}) the matrix has non-zero determinant as claimed.

\smallskip
\noindent
\emph{Case 2: Exactly one $1$ appears in $\mathbf i_1$, so at least one $1$ appears in $\mathbf i_2$.} Again all components of $\mathbf i_1$ must be distinct, so let $a\ne 1$ be one of these. 

Then considering  the row indices
$$\mathbf i_1,\text{ and for each $i\in[k]$, }  \mathbf j_i=(1,1,\dots,1,a,a,i),  $$
and the column indices
$$\mathbf i_2,\text{ and for each $i\in[k]$, } \mathbf k_i=(1,1,\dots,1,a,i),$$
we obtain the needed submatrix.

\smallskip
\noindent
\emph{Case 3: At least two $1$s appear in $\mathbf i_1$,  and
$\mathbf i_2$ has at least one component $a\ne 1$.}
Let $b\ne a $ denote any other component of $\mathbf i_2$ (so $b=1$ is possible).
Then considering  the row indices
$$\mathbf i_1,\text{ and for each $i\in[k]$, } \mathbf j_i=(b,b,\dots,b,a,i),$$
and the column indices
$$\mathbf i_2,\text{ and for each $i\in[k]$, }   \mathbf k_i=(a,a,\dots,a,i),$$
we obtain the needed submatrix.

\smallskip
\noindent
\emph{Case 4: At least two $1$s appear in $\mathbf i_1$,  and
$\mathbf i_2$ has all components $1$.}
Since we are assuming $\mathbf i$ is not constant, $\mathbf i_1$ must have some component $a\ne 1$.
Then considering  the row indices
$$\mathbf i_1,\text{ and for each $i\in[k]$, }   \mathbf j_i=(1,1,\dots,1,a,a,i), $$
and the column indices
$$\mathbf i_2,\text{ and for each $i\in[k]$, }   \mathbf k_i=(1,1,\dots,1,a,i), $$
we obtain the needed submatrix.

\smallskip

At this point all  $q_\mathbf i$ for all non-constant patterns $\mathbf i$ with at least
one repeated component are known. We next use these to determine $q_\mathbf i$ for a constant pattern $\mathbf i$, which we may assume is all $1$s.

\smallskip
\noindent
\emph{Case 5: All components of  $\mathbf i_1$  and
$\mathbf i_2$ are $1$s.}
Considering  the row indices
$$\mathbf i_1,\text{ and for each $i\in[k]$, }  \mathbf j_i=(1,1,\dots,1,2,i),$$
and the column indices
$$\mathbf i_2,\text{ and for each $i\in[k]$, }   \mathbf k_i=(1,1,\dots,1,2,i),$$
we obtain a submatrix all of whose entries except $q_\mathbf i$ are already known.
The non-singularity of the relevant $k\times k$ minor is again shown as in Case 1.

\smallskip

A final case shows we can determine the remaining $q_\mathbf i$, which have no repeated components

\smallskip
\noindent
\emph{Case 6: No components of  $\mathbf i$  are repeated.}
Considering  the row indices
$$\mathbf i_1,\text{ and for each $i\in[k]$, }  \mathbf j_i=(1,1,\dots,1,i),$$
and the column indices
$$\mathbf i_2,\text{ and for each $i\in[k]$, }   \mathbf k_i=(1,1,\dots,1,i),$$
we obtain a submatrix all of whose entries except $q_\mathbf i$ are already known, whose
relevant $k\times k$ minor is similarly shown to be non-singular .
\end{proof}

\bibliographystyle{elsarticle-harv}

\bibliography{mkconsist}
\end{document}